\documentclass[preprint]{revtex4-1}
\usepackage[pdftex]{graphicx}
\usepackage{epstopdf}
\expandafter\let\csname equation*\endcsname\relax
\expandafter\let\csname endequation*\endcsname\relax
\usepackage{amsmath}
\usepackage{amsfonts}
\usepackage{amssymb}
\usepackage{makeidx}

\newcommand{\beq}{\begin{equation}} \newcommand{\eeq}{\end{equation}}
\newcommand{\bea}{\begin{eqnarray}} \newcommand{\eea}{\end{eqnarray}}
\newcommand{\bear}{\begin{eqnarray*}} \newcommand{\eear}{\end{eqnarray*}}
\newcommand{\lb}{\label} 
\newcommand{\rf}[1]{(\ref{#1})}   

\newtheorem{theorem}{Theorem}[section]
\newtheorem{lemma}[theorem]{Lemma}

\newtheorem{definition}{Definition}
\newtheorem{remark}{Remark}

\newenvironment{proof}[1][Proof]{\begin{trivlist}
\item[\hskip \labelsep {\bfseries #1}]}{\end{trivlist}}

\begin{document}

\title {The action principle for dissipative systems}

\author{Matheus J. Lazo}
\email{matheuslazo@furg.br}
\author{Cesar E. Krumreich}

\address{Instituto de Matem\'atica, Estat\'\i stica e F\'\i sica - FURG, Rio Grande, RS, Brazil.}

\begin{abstract}

In the present work we redefine and generalize the action principle for dissipative systems proposed by Riewe by fixing the mathematical inconsistencies present in the original approach. In order to formulate a quadratic Lagrangian for non-conservative systems, the Lagrangian functions proposed depend on mixed integer order and fractional order derivatives. As examples, we formulate a quadratic Lagrangian for a particle under a frictional force proportional to the velocity, and to the classical problem of an accelerated point charge.

%{\it PACS}: 45.20.Jj; 45.10.Db; 02.30.Xx; 41.60.-m; 05.45.Ac

%\keywords{Action Principle, Lagrangians for dissipative systems, fractional calculus}

\end{abstract}

\maketitle

%%%%%%%%%%%%%%%%%%%%%%%%%%%%%%%%%%%%%%%%%%%%%%%%%%%%%%%%%%%%%%%%%%%%%%%%%%%%%%%%%%%%%%%%%%%%

\section{Introduction}

Since the introduction of the action principle in its mature formulation by Euler, Hamilton and Lagrange, it is well known that the equation of motion of dissipative linear dynamical systems with constant coefficients can not be obtained by variational principle. A rigorous proof of the failure of the action principle for non-conservative systems was given in 1931 by Bauer \cite{bauer} when he proved the impossibility to obtain a dissipation term proportional to the first order time derivative in the equation of motion from a variational principle. Over the last century, several methods were developed in order to deal with this failure. Examples include time dependent Lagrangians \cite{Stevens}, and the Bateman approach by introducing auxiliary coordinates that describe the reverse-time system \cite{Morse} (see \cite{VujaJones} for a review). Unfortunately, all these approaches give us non-physical Lagrangians in the sense they provide non-physical relations for the momentum and Hamiltonian of the system (see \cite{Riewe} for a detailed discussion). 

Only recently, by exploring a loophole in Bauer's proof (in his proof Bauer assumed that all derivatives were integer order), Riewe \cite{Riewe} showed that quadratic Lagrangians involving fractional time derivatives leads to equation of motion with non-conservative forces such as friction. Furthermore, the quadratic Lagrangian with fractional derivatives proposed by Riewe have the advantage over other approaches of being physical in the sense that it provides meaningful relations for momentum and Hamiltonian \cite{Riewe}. The Riewe's approach follows from the observation that if the Lagrangian contains a term proportional to $\left(d^n x/dt^n\right)^2$, then the Euler-Lagrange equation will have a corresponding term proportional to $d^{2n}x/dt^{2n}$. Hence a frictional force proportional to the velocity $dx/dt$ should follow directly from a Lagrangian containing a quadratic term proportional to the fractional derivative $\left(d^{\frac{1}{2}}x/dt^{\frac{1}{2}}\right)^2$. 

However, since fractional derivatives are in general non-local operators with algebraic properties different from usual derivatives \cite{OldhamSpanier,SATM,Kilbas,Hilfer,Magin,SKM,Diethelm}, in order to obtain the equation of motion for dissipative systems in the Riewe's approach, we should redefine the action principle by taking a given limit, make some not well defined approximations, and introduce a complex Lagrangian \cite{Riewe}. Despite the mathematical inconsistency in these approximations, the idea proposed by Riewe proves to be so novel and interesting that it transcends the problem that it was originally designed to solve and becomes an area of study in its own right. Nowadays, the fractional calculus of variations is being developed as a tool to study a wide variety of problems \cite{MalinowskaTorres}.

In the present work we redefine and generalize the action principle for dissipative systems by fixing the mathematical inconsistencies present in the Riewe's formulation. Our approach also has the advantage that the Lagrangian can be a real valued function with a direct physical meaningful. Furthermore we also generalize the Riewe's approach in order to formulate Lagrangians for higher order dissipative systems, like the classical problem of an accelerated point charge \cite{landau}.

Our paper is organized as follows. In section $2$ we review the Riemann-Liouville and Caputo fractional calculus. The action principle for dissipative systems, with Lagrangian functions with Caputo derivatives, is introduced in section $3$, and generalized to higher-order systems in section $4$. Finally, the conclusion is presented in section $5$.

%%%%%%%%%%%%%%%%%%%%%%%%%%%%%%%%%%%%%%%%%%%%%%%%%%%%%%%%%%%%%%%%%%%%%%%%%%%%%%%%%%%%%%%%%%%%

\section{The Riemann-Liouville and Caputo Fractional Calculus}

The fractional calculus of derivative and integration of non-integers orders started more than three centuries ago with l'H\^opital and Leibniz when a derivative of order $\frac{1}{2}$ was suggested \cite{OldhamSpanier}. This subject was  also considered by several mathematicians as Euler, Laplace, Liouville, Grunwald, Letnikov, Riemann and others up to nowadays. Although the fractional calculus is almost as old as the usual integer order calculus, only in the last three decades it has gained more attention due to its applications in various fields of science (see \cite{SATM,Kilbas,Hilfer,Magin} for a review). Actually, there are several definitions of fractional order derivatives. Theses definitions include the Riemann-Liouville, Caputo, Riesz, Weyl,  Grunwald-Letnikov, etc. (see \cite{OldhamSpanier,SATM,Kilbas,Hilfer,Magin,SKM,Diethelm} for a review). In this section we review some definitions and properties of the Riemann-Liouville and Caputo fractional calculus.

Despite we have many different approaches to fractional calculus, several known formulations are somehow connected with the analytic continuation of Cauchy formula for $n$-fold integration
\beq
\lb{a2}
\begin{split}
\int_{a}^t x(\tilde{t})(d\tilde{t})^{n} &= \int_{a}^t\int_{a}^{t_{n}}\int_{a}^{t_{n-1}}\cdots \int_{a}^{t_3}\int_{a}^{t_2} x(t_1)dt_1dt_2\cdots dt_{n-1}dt_{n} \\
&= \frac{1}{\Gamma(n)}\int_{a}^t \frac{x(u)}{(t-u)^{1-n}}du \;\;\;\;\; (n\in \mathbb{N}),
\end{split}
\eeq
where $\Gamma$ is the Euler gamma function. The proof of Cauchy formula can be found in several textbooks (for example, it can be found in \cite{OldhamSpanier}). The analytic continuation of \rf{a2} gives us a definition for an integration of non-integer (or fractional) order. This fractional order integration is the building bloc of the Riemann-Liouville and Caputo calculus, the two most popular formulations of fractional calculus, as well as several other approaches to fractional calculus \cite{OldhamSpanier,SATM,Kilbas,Hilfer,Magin,SKM,Diethelm}. The fractional integration obtained from \rf{a2} are historically called Riemann-Liouville left and right fractional integrals: 
\begin{definition} Let $\alpha \in \mathbb{R}_+$. The operators ${_a J^{\alpha}_t}$ and ${_t J^{\alpha}_b}$ defined on $L_1[a,b]$ by
\beq
\lb{a3}
{_a J^{\alpha}_t} x(t) =\frac{1}{\Gamma(\alpha)}\int_{a}^t \frac{x(u)}{(t-u)^{1-\alpha}}du 
\eeq
and
\beq
\lb{a4}
{_t J^{\alpha}_b} x(t) =\frac{1}{\Gamma(\alpha)}\int_t^b \frac{x(u)}{(u-t)^{1-\alpha}}du ,
\eeq
with $a<b$ and $a,b\in \mathbb{R}$, are called left and the right fractional Riemann-Liouville integrals of order $\alpha$, respectively.
\end{definition}

For integer $\alpha$ the fractional Riemann-Liouville integrals \rf{a3} and \rf{a4} coincide with the usual integer order $n$-fold integration \rf{a2}. Moreover, from the definitions \rf{a3} and \rf{a4} it is easy to see that the Riemann-Liouville fractional integrals converge for any integrable function $x$ if $\alpha>1$. Furthermore, it is possible to proof the convergence of \rf{a3} and \rf{a4} for $x\in L_1[a,b]$ even when $0<\alpha<1$ \cite{Diethelm}.

It can be directly verified that the fractional Riemann-Liouville integrals \rf{a3} and \rf{a4} of the power functions $(t-a)^{\beta}$ and $(b-t)^{\beta}$ yield power functions of the same form \cite{Kilbas,Diethelm}:
\begin{remark}
Let $\alpha>0$ and $\beta>-1$ with $\alpha,\beta \in \mathbb{R}$. Then
\beq
\lb{a41a}
{_a J^{\alpha}_t} (t-a)^{\beta}=\frac{\Gamma(\beta+1)}{\Gamma(\beta+\alpha+1)} (t-a)^{\beta+\alpha},
\eeq
and
\beq
\lb{a41b}
{_t J^{\alpha}_b} (b-t)^{\beta}=\frac{\Gamma(\beta+1)}{\Gamma(\beta+\alpha+1)} (b-t)^{\beta+\alpha}. 
\eeq
\end{remark}

The integration operators ${_a J^{\alpha}_t}$ and ${_t J^{\alpha}_b}$ play a fundamental role in the definition of fractional Riemann-Liouville and Caputo calculus. In order to define the Riemann-Liouville derivatives, we recall that for positive integers $n>m$ it follows the identity $D^m_t x(t)=D^{n}_t {_aJ^{n-m}_t x(t)}$, where $D^m_t$ is an ordinary derivative of integer order $m$. 
\begin{definition}[Riemann-Liouville]
The left and the right Riemann-Liouville fractional derivative of order $\alpha >0$ ($\alpha\in \mathbb{R}$) are defined on $L_1[a,b]$, respectively, by ${_a D^{\alpha}_t} x(t) := D^{n}_t {_a J^{n-\alpha}_t} x(t)$ and ${_t D^{\alpha}_b} x(t):=(-1)^nD^{n}_t{_t J^{n-\alpha}_b} x(t)$ with $n=[\alpha]+1$, namely
\beq
\lb{a5}
{_a D^{\alpha}_t} x(t)=\frac{1}{\Gamma(n-\alpha)}\frac{d^n}{dt^n}\int_{a}^t \frac{x(u)}{(t-u)^{1+\alpha-n}}du 
\eeq
and
\beq
\lb{a6}
{_t D^{\alpha}_b} x(t)=\frac{(-1)^n}{\Gamma(n-\alpha)}\frac{d^n}{dt^n}\int_{t}^b \frac{x(u)}{(u-t)^{1+\alpha-n}}du,
\eeq
where $\frac{d^n}{dt^n}$ stands for ordinary derivatives of integer order $n$. 
\end{definition}
On the other hand, the Caputo fractional derivatives are defined by inverting the order between derivatives and integrations
\begin{definition}[Caputo]
The left and the right Caputo fractional derivatives of order $\alpha\in \mathbb{R}_+$ are defined on $C^n[a,b]$, respectively,
by ${_a^C D}^{\alpha}_t x(t) := {_aJ}^{n-\alpha}_t D^{n}_t x(t)$
and ${_t^C}D^{\alpha}_b x(t) := (-1)^n _tJ^{n-\alpha}_b
D^{n}_t x(t)$ with $n=[\alpha]+1$; that is,
\begin{equation}
\label{a7}
{_a^C D}^{\alpha}_t x(t) := \frac{1}{\Gamma(n-\alpha)}
\int_{a}^t \frac{x^{(n)}(u)}{(t-u)^{1+\alpha-n}}du
\end{equation}
and
\begin{equation}
\label{a8}
{_t^C D}^{\alpha}_b x(t)
:= \frac{(-1)^n}{\Gamma(n-\alpha)}\int_{t}^b
\frac{x^{(n)}(u)}{(u-t)^{1+\alpha-n}}du,
\end{equation}
where $a \le t \le b$ and $x^{(n)}(u)=\frac{d^n x(u)}{du^n}$
is the ordinary derivative of integer order $n$.
\end{definition}

An important consequence of definitions \eqref{a5}--\eqref{a8} is that the Riemann-Liouville and Caputo fractional derivatives are non-local operators. The left (right) differ-integration operator \eqref{a5} and \eqref{a7} (\eqref{a6} and \eqref{a8}) depends on the values of the function at left (right) of $t$, i.e. $a\leq u \leq t$ ($t\leq u \leq b$). On the other hand, it is important to note that when $\alpha$ is an integer, the Riemann-Liouville fractional derivatives \eqref{a5} and \eqref{a6} reduce to ordinary derivatives of order $\alpha$. On the other hand, in that case, the Caputo derivatives \eqref{a7} and \eqref{a8} differ from integer order ones by a polynomial of order $\alpha -1$ {\rm \cite{Kilbas,Diethelm}. Furthermore, it can be easily verified that the Riemann-Liouville and Caputo derivatives are connected with each other by the following relations:
\beq
\lb{a8a}
{_a^C D}^{\alpha}_t x(t)={_a D}^{\alpha}_t x(t)-\sum_{k=0}^{n-1}\frac{x^{(k)}(a)}{\Gamma(k-\alpha+1)}(t-a)^{k-\alpha},
\eeq
and
\beq
\lb{a8b}
{_t^C D}^{\alpha}_b x(t)={_t D}^{\alpha}_b x(t)-\sum_{k=0}^{n-1}\frac{x^{(k)}(b)}{\Gamma(k-\alpha+1)}(b-t)^{k-\alpha},
\eeq
where $n=[a]+1$.

For power functions $(t-a)^{\beta}$ and $(b-t)^{\beta}$ the Riemann-Liouville and Caputo fractional derivatives yield \cite{Kilbas,Diethelm}:
\begin{remark}
Let $\alpha>0$ and $\beta>-1$ with $\alpha,\beta \in \mathbb{R}$. We have for the Riemann-Liouville derivatives
\beq
\lb{a81}
{_a D}^{\alpha}_t (t-a)^{\beta}=\frac{\Gamma(\beta+1)}{\Gamma(\beta-\alpha+1)} (t-a)^{\beta-\alpha},
\eeq
and
\beq
\lb{a82}
{_t D}^{\alpha}_b (b-t)^{\beta}=\frac{\Gamma(\beta+1)}{\Gamma(\beta-\alpha+1)} (b-t)^{\beta-\alpha},
\eeq
if $\alpha-\beta \notin \mathbb{N}$, and zero when $\alpha-\beta \in \mathbb{N}$. For the Caputo derivatives we have, similarly,
\beq
\lb{a83}
{_a^C D}^{\alpha}_t (t-a)^{\beta}=\frac{\Gamma(\beta+1)}{\Gamma(\beta-\alpha+1)} (t-a)^{\beta-\alpha},
\eeq
and
\beq
\lb{a84}
{_t^C D}^{\alpha}_b (b-t)^{\beta}=\frac{\Gamma(\beta+1)}{\Gamma(\beta-\alpha+1)} (b-t)^{\beta-\alpha}. 
\eeq
when $\beta\neq 0, 1, 2,...,[\alpha]$, and zero if $\beta= 0,1, 2,...,[\alpha]$.

In particular, if $\beta=0$ the Riemann-Liouville fractional derivatives of a constant are, in general, not equal to zero:
\beq
\lb{a85}
{_a D}^{\alpha}_t 1=\frac{(t-a)^{-\alpha}}{\Gamma(1-\alpha)},\;\;\; \mbox{and}\;\;\; {_t D}^{\alpha}_b 1=\frac{(b-t)^{-\alpha}}{\Gamma(1-\alpha)},
\eeq
while the Caputo fractional derivatives of a constant is always zero.

\end{remark}

In addition to these definitions, in the present work we make use of the following property
\beq
\lb{a18}
{_a D_t^{\frac{1}{2}}} {_a^C D_t^{\frac{1}{2}}} x(t)=\frac{d}{dt}{_a J^{\frac{1}{2}}_t}{_a J^{\frac{1}{2}}_t}\frac{d}{dt}x(t)=\frac{d}{dt}{_a J^{1}_t}\frac{d}{dt}x(t)=\frac{d}{dt}x(t)
\eeq
that follows from the general property ${_a J^{\alpha}_t}{_a J^{\beta}_t}={_a J^{\alpha+\beta}_t}$ for any differentiable function $x(t)$ (see, e.g., \cite{SKM,Diethelm}), and the integration by parts
\begin{theorem}[Integration by parts --- see, e.g., \cite{Agrawal2}]
\label{thm:ml:03}
Let $0<\alpha<1$ and $x$ be a differentiable function in $[a,b]$ with $x(a)=x(b)=0$. For any function $y \in L_1([a,b])$ one has
\begin{equation}
\label{a15}
\int_{a}^{b} y(t) {_a^C D_t^{\alpha}} x(t)dt
= \int_a^b x(t) {_t D_b^{\alpha}} y(t)dt
\end{equation}
and
\begin{equation}
\label{a16}
\int_{a}^{b} y(t)  {_t^C D_b^\alpha} x(t)dt
=\int_a^b x(t) {_a D_t^\alpha} y(t) dt.
\end{equation}
\end{theorem}
It is important to notice that the formulas of integration by parts \eqref{a15} and \eqref{a16} relate Caputo left (right) derivatives to Riemann-Liouville right (left) derivatives.

%%%%%%%%%%%%%%%%%%%%%%%%%%%%%%%%%%%%%%%%%%%%%%%%%%%%%%%%%%%%%%%%%%%%%%%%%%%%%%%%%%%%%%%%%%

\section{The Action Principle for dissipative systems}

In this section we propose an action principle, or Hamilton's principle, for Lagrangian functions depending on fractional derivatives in order to fix the mathematical inconsistencies present in Riewe's approach \cite{Riewe}. 

The classical action principle states that a system moves from a given configuration to another, on the time interval $[a,b]$, in such a way that the variation of the action integral $S=\int_a^b Ldt$ between the path taken and a neighboring virtual path is zero, where $L$ is the Lagrangian function. In the language of calculus of variations, the action principle mean that the first variation of the functional $S$ should be zero, namely
\beq
\lb{b1}
\delta S=\delta \int_a^b Ldt=0,
\eeq
and the system path is given by the solution of the Euler-Lagrange equation obtained from \rf{b1}. 

For any positive integer $n$ it is a simple exercise to show that if the Lagrangian contains a term proportional to $\left(d^n x/dt^n\right)^2$ then the Euler-Lagrange equation will have a corresponding term proportional to $d^{2n}x/dt^{2n}$. As a consequence, it is not possible to formulate quadratic Lagrangian function for dissipative systems containing odd order time derivatives in the equation of motion, like a frictional force proportional to the velocity $\dot{x}=dx/dt$. In order to bypass this difficulty, Riewe proposed a Lagrangian for a particle under this kind of frictional force by introducing a quadratic term containing a Riemann-Liouville derivative of order $\alpha=1/2$ \cite{Riewe}:
\beq
\lb{b4}
L\left(x,\dot{x},{_t D^{\frac{1}{2}}_b} x\right)=\frac{1}{2}m\left(\dot{x}\right)^2-U(x)+i\frac{\gamma}{2}\left({_t D^{\frac{1}{2}}_b} x\right)^2,
\eeq
where the three terms in \rf{b4} represent the kinetic energy, potential energy, and the fractional linear friction energy, respectively. The action Riewe consider is then given by
\beq
\lb{b2}
S=\int_a^b {L}\left(x,\dot{x},{_t D^{\alpha}_b} x\right) dt,
\eeq
for which the path taken $x(t)$ should satisfy the fractional Euler-Lagrange equation \cite{Riewe}
\beq
\lb{b3}
\frac{\partial {L}}{\partial x}-\frac{d}{dt}\frac{\partial {L}}{\partial \left(\dot{x}\right)}+{_a D^{\alpha}_t}\frac{\partial {L}}{\partial\left({_t D^{\alpha}_b} x\right)}=0.
\eeq
It is important to notice that while the Lagrangian function in \rf{b2} contains only right Riemann-Liouville derivative, the Euler-Lagrange equation \rf{b3} contain both right and left derivatives. The Euler-Lagrange equation \rf{b3} give us for \rf{b4} the following fractional differential equation
\beq
\lb{b5}
m \ddot{x}-i\gamma{_a D^{\frac{1}{2}}_t} {_t D^{\frac{1}{2}}_b}x=F(x),
\eeq
where $\ddot{x}=d^2x/dt^2$ is the acceleration and $F(x)=-\frac{d}{dx}U(x)$ is the external force. Since in general ${_a D^{\frac{1}{2}}_t} {_t D^{\frac{1}{2}}_b}x\neq \dot{x}$, Riewe proposed a modification to the action principle by stating that the equation of motion is obtained after taking the limit $a\rightarrow b$ in the Euler-Lagrange equation \rf{b3}. In addition Riewe also made the approximation $i {_a D_t^{\frac{1}{2}}} f(t) \approx {_t D_b^{\frac{1}{2}}}f(t)$, and used the relation ${_t D_b^{\frac{1}{2}}}{_t D_b^{\frac{1}{2}}}f(t)=\frac{d}{dt}f(t)$, obtaining \cite{Riewe}
\beq
\lb{b6}
\lim_{a\rightarrow b}\left(m \ddot{x}-i\gamma{_a D^{\frac{1}{2}}_t} {_t D^{\frac{1}{2}}_b}x\right)=\lim_{a\rightarrow b}\left(m \ddot{x}+\gamma{_t D^{\frac{1}{2}}_b} {_t D^{\frac{1}{2}}_b}x\right)=m \ddot{x}+\gamma \dot{x}=F(x).
\eeq

It is important to emphasize that the condition $a\rightarrow b$ applied to the action principle does not imply any restrictions for conservative systems, since in this case the path $x(t)$ is the action's extremal for any time interval $[a,b]$, even when $a\rightarrow b$. However, the Riewe's approach displays two mathematical problems. First, if $x(b)\neq 0$ we have ${_t D_b^{\frac{1}{2}}}{_t D_b^{\frac{1}{2}}}x$ defined only in $[a,b)$ and not in $[a,b]$ and, beyond that, ${_t D_b^{\frac{1}{2}}}{_t D_b^{\frac{1}{2}}}x=\dot{x}$ is not always valid \cite{OldhamSpanier,SATM,Kilbas,Hilfer,Magin,SKM,Diethelm}. Second and most important, the approximation $i {_a D_t^{\frac{1}{2}}} x(t) \approx {_t D_b^{\frac{1}{2}}}x(t)$ is actually inconsistent. In addition to both left and right fractional derivatives of $x(t)$ being real valued functions, the ratio ${_a D_t^{\frac{1}{2}}} x(t) / {_t D_b^{\frac{1}{2}}}x(t)$ in the limit $a\rightarrow b$ is indefinite (see Appendix A). Finally, the fractional linear friction energy defined in \rf{b4} has the additional drawback of not display physical reality due to the imaginary number in \rf{b4} and, most important, since it diverges in the limit $a\rightarrow b$ if $x(b)\neq 0$ (see Appendix A).

In order to fix these inconsistencies we propose two modifications to the action principle for dissipative systems. Furthermore, our approach has the advantage of allowing a real valued Lagrangian with direct physical interpretation. The first mathematical limitation, as well as the divergence of the linear friction energy, can be easily removed if we replace the Riemann-Liouville derivative in the Lagrangian by a Caputo derivative (see discussion at the end of this section). Furthermore, we show that the second inconsistency can be fixed by specifying the way we takes the limit $a\rightarrow b$. To proof the last statement, let we analyze the ratio ${_a^C D_t^{\frac{1}{2}}} x(t) / {_t^C D_b^{\frac{1}{2}}}x(t)$ in the limit $a\rightarrow b$. If we consider $x(t)$ at least $C^1[a,b]$, there are real numbers $a<c_a<t$ and $t<c_b<b$ such that $x(t)=x(a)+\dot{x}(c_a)(t-a)$ in $[a,t]$ and $x(t)=x(b)-\dot{x}(c_b)(b-t)$ in $[t,b]$. Then, if $\dot{x}(b)\neq 0$,
\beq
\lb{ta4b}
\lim_{a\rightarrow b^-} \frac{{_a^C D_t^{\alpha}} x(t)}{{_t^C D_b^{\alpha}}x(t)}=-\lim_{a\rightarrow b^-} \frac{\dot{x}(c_a)}{\dot{x}(c_b)}\left(\frac{t-a}{b-t}\right)^{1-\alpha}=-\lim_{a\rightarrow b^-} \left(\frac{t-a}{b-t}\right)^{1-\alpha},
\eeq
since $\lim_{a\rightarrow b^-}\dot{x}(c_a)/\dot{x}(c_b)=1$. Note that the limit \rf{ta4b} is indefinite for arbitrary $t\in [a,b]$. For example, by choosing $t=a+s(b-a)$ with $0<s<1$ we obtain 
\beq
\lb{ta6}
\lim_{a\rightarrow b^-} \frac{{_a^C D_t^{\alpha}} x(t)}{{_t^C D_b^{\alpha}}x(t)}=-
\lim_{a\rightarrow b^-} \left(\frac{s}{1-s}\right)^{1-\alpha}=-\left(\frac{s}{1-s}\right)^{1-\alpha},
\eeq
that depends on the value $0<s<1$. In conclusion, in order to get a definite limit for the ratio ${_a^C D_t^{\frac{1}{2}}} x(t) / {_t^C D_b^{\frac{1}{2}}}x(t)$ we should specify the way $t$ tends to $b$ in the limit $a\rightarrow b$. We can now state the following Lemma:
\begin{lemma} Let $x(t)\in C^1[a,b]$ with $\dot{x}(b)\neq 0$, and let $\alpha\in \mathbb{R}$ with $0<\alpha<1$. Then we have, for a given real number $0<s<1$,
\beq
\lb{ta1}
\lim_{a\rightarrow b^-} \frac{{_a^C D_t^{\alpha}} x(t)}{{_t^C D_b^{\alpha}}x(t)}=-\left(\frac{s}{1-s}\right)^{1-\alpha} \;\;\; \mbox{if} \;\;\; t=a+s(b-a).
\eeq
\end{lemma}
In particular, from this Lemma we can conclude that for $s=\frac{1}{2}$ (when $t$ is the midpoint of $[a,b]$) we can approximate
\beq
\lb{b7}
{_t^C D_b^{\alpha}}x(t)\approx -{_a^C D_t^{\alpha}}x(t) \;\;\;\;\; \mbox{for} \;\;\;\;\; a-b<<1.
\eeq

We can now formulate an action principle for dissipative systems free from the problems found in Riewe's approach. The action principle we propose states that the equation of motion for dissipative systems is obtained by taking the limit $a\rightarrow b$ with $t=a+(b-a)/2=(a+b)/2$ in the extremal of the action
\beq
\lb{b8}
S=\int_a^b {L}\left(x,\dot{x},{_t^C D^{\alpha}_b} x\right) dt,
\eeq
that satisfy the fractional Euler-Lagrange equation (see \cite{ELCaputo} and references therein):
\beq
\lb{b9}
\frac{\partial {L}}{\partial x}-\frac{d}{dt}\frac{\partial {L}}{\partial \left(\dot{x}\right)}+{_a D^{\alpha}_t}\frac{\partial {L}}{\partial\left({_t^C D^{\alpha}_b} x\right)}=0.
\eeq

Note that in the definition of the action principle we choose $t=(a+b)/2$ ($s=1/2$) as the midpoint of the time interval $[a,b]$. Actually, this choice do not imply any restriction since for $s\neq 1/2$ we can redefine the Lagrangian by introducing a constant in order to absorb the factor $(s/(1-s))^{1-\alpha}$ that should appear in the approximation \rf{b7} due to \rf{ta1}. Furthermore, the condition $\dot{x}(b)\neq 0$, in both \rf{ta1} and \rf{b7}, also do not impose restriction to our problem since the physical path $x(t)$ is at least a $C^2$ function, and we can take the analytic continuation of our solution to the points where $\dot{x}(b)= 0$. Finally, in order to display the absence of mathematical inconsistencies in our approach, and in order to show that our method provides us with physical Lagrangians, let us consider the simple problem of a particle under a frictional force proportional to velocity. 

\subsection{A quadratic Lagrangian for the linear friction problem}

The action principle we propose enables us to formulate a quadratic Lagrangian for a particle under a frictional force proportional to the velocity as
\beq
\lb{b10}
L\left(x,\dot{x},{_t^C D^{\frac{1}{2}}_b} x\right)=\frac{1}{2}m\left(\dot{x}\right)^2-U(x)+\frac{\gamma}{2}\left({_t^C D^{\frac{1}{2}}_b} x\right)^2,
\eeq
where the three terms in \rf{b10} represent the kinetic energy, potential energy, and the fractional linear friction energy, respectively. Note that different from the Riewe's Lagrangian \rf{b4} our Lagrangian \rf{b10} is a real function with a linear friction energy physically meaningful. Since the equation of motion is obtained in the limit $a\rightarrow b$, if we consider the last term in \rf{b10} up to first order in $\Delta t=b-a$ we get:
\beq
\lb{b11}
\frac{\gamma}{2}\left({_t^C D^{\frac{1}{2}}_b} x\right)^2\approx \frac{\gamma}{2}\left(\frac{\Gamma(1)}{\Gamma(\frac{3}{2})}\right)^2\left(\dot{x}\right)^2\Delta t \approx \frac{2}{\pi} \gamma \dot{x} \Delta x,
\eeq
that coincide, apart from the multiplicative constant $2/\pi$, with the work from the frictional force $\gamma \dot{x}$ in the displacement $\Delta x \approx \dot{x}\Delta t$. This additional constant is a consequence of the use of fractional derivatives in the Lagrangian and do not appears in the equation of motion after we apply the action principle. Furthermore, the Lagrangian \rf{b10} is physical in the sense it provide us with physically meaningful relations for the momentum and the Hamiltonian. If we define the canonical variables
\beq
\lb{b12}
q_{1}=\dot{x}, \;\;\; q_{\frac{1}{2}}={_t^C D^{\frac{1}{2}}_b} x,
\eeq
and
\beq
\lb{b13}
p_1=\frac{\partial L}{\partial q_1}=m\dot{x}, \;\;\; p_{\frac{1}{2}}=\frac{\partial L}{\partial q_{\frac{1}{2}}}=\gamma{_t^C D^{\frac{1}{2}}_b} x,
\eeq
we obtain the Hamiltonian
\beq
\lb{b14}
H=q_1p_1+q_{\frac{1}{2}}p_{\frac{1}{2}}-L=\frac{1}{2}m\left(\dot{x}\right)^2+U(x)+\frac{\gamma}{2}\left({_t^C D^{\frac{1}{2}}_b} x\right)^2.
\eeq
From \rf{b13} and \rf{b14} we can see that the Lagrangian \rf{b10} is physical in the sense it provides us a correct relation for the momentum $p_1=m\dot{x}$, and a physically meaningful Hamiltonian (it is the sum of all energies). Furthermore, the additional fractional momentum $p_{\frac{1}{2}}=\gamma{_t^C D^{\frac{1}{2}}_b} x$ goes to zero when we takes the limit $a\rightarrow b$.

Finally, the equation of motion for the particle is obtained by inserting our Lagrangian \rf{b10} into the Euler-Lagrange equation \rf{b9},
\beq
\lb{b15}
m \ddot{x}-\gamma{_a D^{\frac{1}{2}}_t} {_t^C D^{\frac{1}{2}}_b}x=F(x),
\eeq
where $F(x)=-\frac{d}{dx}U(x)$ is the external force. By taking the limit $a\rightarrow b$ with $t=(a+b)/2$ and using the approximation \rf{b7} and the relation \rf{a18} we obtain
\beq
\lb{b16}
m \ddot{x}+\gamma \dot{x}=F(x).
\eeq

%%%%%%%%%%%%%%%%%%%%%%%%%%%%%%%%%%%%%%%%%%%%%%%%%%%%%%%%%%%%%%%%%%%%%%%%%%%%%%%%%%%%%%%%%%

\section{The action principle for higher-order dissipative systems}

In this section we generalize our action principle in order to investigate new kinds of quadratic Lagrangian suitable to study higher-order systems. Lagrangians with higher-order derivatives have gained attention in the last three decades due to applications in several fields. Higher-order derivatives appear naturally as corrections to renormalizable gauge field theories \cite{FaddeevSlavnov}, gravity \cite{rgravity}, cosmic strings \cite{BirellDavies},  dark energy physics \cite{dark}, electrodynamics \cite{electro}, and other problems \cite{other}. Furthermore, higher-order effective Lagrangians have been intensely studied in the literature in order to parametrize possible deviations of the electroweak interactions from the standard model \cite{Higgs}. Therefore, the study of Lagrangians with higher-order derivatives are important for several physical problems, at least from the phenomenological point of view. Actually, we need not go so far to find physical systems with higher-order derivatives. As an example, we have the classical problem of an accelerated point charge where we have a radiation recoil force proportional to $\dddot{x}=d^3x/dt^3$. 

In order to deal with Lagrangian functions with fractional derivatives and higher-order integer derivatives we introduced the following Theorem
\begin{theorem}
Let $\alpha_{j}\in \mathbb{R}$ with $0<\alpha_{j}<1$ ($j=0,1,...,n$), and $S$ be an action of the form
\begin{equation}
\label{t1}
\begin{split}
S=\int_a^b {L} &\left(t,x,\frac{dx}{dt},...,\frac{d^nx}{dt^n},{_a^C D^{\alpha_0}_t} x,{_a^C D^{\alpha_1}_t} \frac{dx}{dt},...,{_a^C D^{\alpha_n}_t} \frac{d^nx}{dt^n},\right. \\
&\quad \quad \quad \quad \quad \quad \quad \quad \quad \quad \left. {_t^C D^{\alpha_0}_b} x,{_t^C D^{\alpha_1}_b} \frac{dx}{dt},...,{_t^C D^{\alpha_n}_b} \frac{d^nx}{dt^n}\right)dt,
\end{split}
\end{equation}
where the function $x \in C^{n+2}[a,b]$ satisfies the fixed boundary conditions $x(a)=x^{(0)}_a$, $x(b)=x^{(0)}_b$, and $\frac{d^{j}x(a)}{dt^{j}}=x^{(j)}_a$, $\frac{d^{j}x(b)}{dt^{j}}=x^{(j)}_b$ with $x^{(j)}_a,x^{(j)}_b \in \mathbb{R}$ for $j=1,...,n-1$. Also let ${L}\in C^{2}[a,b]\times \mathbb{R}^{3n+3}$. Then, the necessary condition for $S$ to possess an extremum is that the function $x$ fulfills the following fractional Euler-Lagrange equation:
\begin{equation}
\label{t2}
\begin{split}
\frac{\partial {L}}{\partial x}&+\sum_{j=1}^n(-1)^j\frac{d^j}{dt^j}\frac{\partial {L}}{\partial \left(\frac{d^jx}{dt^t}\right)}\\
&+\sum_{j=0}^{n}(-1)^j\frac{d^j}{dt^j} \left( {_t D^{\alpha_j}_b}\frac{\partial {L}}{\partial\left({_a^C D^{\alpha_j}_t} \frac{d^jx}{dt^j}\right)}+{_a D^{\alpha_j}_t}\frac{\partial {L}}{\partial\left({_t^C D^{\alpha_j}_b} \frac{d^jx}{dt^j}\right)}\right)=0,
\end{split}
\end{equation}
where $\frac{d^0}{dt^0}\equiv 1$.
\end{theorem}
\begin{proof}
In order to develop the necessary conditions for the extremum of the action \rf{t1}, we define a family of functions $x$ (weak variations)
\beq
\lb{p1}
x=x^*+\varepsilon \eta,
\eeq
where $x^*$ is the desired real function that satisfies the extremum of \rf{t1}, $\varepsilon \in \mathbb{R}$ is a constant, and the function $\eta$ defined in $[a,b]$ satisfies the boundary conditions
\beq
\lb{p2}
\eta(a)=\eta(b)=0, \;\;\; \frac{d^j\eta(a)}{dt^j}=\frac{d^j\eta(b)}{dt^j}=0 \;\;\; (j=1,...,n-1).
\eeq
The condition for the extremum is obtained when the first G\^ateaux variation is zero
\beq
\lb{p3}
\begin{split}
\delta S&=\lim_{\varepsilon \rightarrow 0} \frac{S[x^*+\varepsilon \eta]-S[x^*]}{\varepsilon}=\int_a^b \left[ \eta \frac{\partial L}{\partial x^*} +\sum_{j=1}^n\frac{d^j \eta}{dt^j} \frac{\partial {L}}{\partial\left(\frac{d^jx^*}{dt^j}\right)} \right.\\
&\left. \;\;\;\;\;\;\;\; +\sum_{j=0}^{n}\left( {_a^C D^{\alpha_j}_t} \frac{d^j \eta}{dt^j} \frac{\partial {L}}{\partial\left({_a^C D^{\alpha_j}_t} \frac{d^jx^*}{dt^j}\right)} + {_t^C D^{\alpha_j}_b} \frac{d^j \eta}{dt^j} \frac{\partial {L}}{\partial\left({_t^C D^{\alpha_j}_b} \frac{d^jx^*}{dt^j}\right)}\right)\right]dx=0. 
\end{split}
\eeq
Finally, by using the formula for integration by part \cite{OldhamSpanier,SATM}, the boundary conditions \rf{p2} and the fundamental lemma of the calculus of variations, we obtain the fractional Euler-Lagrange equations \rf{t2}. 
\end{proof}

It is important to notice that since ${_a^C D^{\alpha}_t} \frac{d^j}{dt^j}={_a^C D^{\alpha+j}_t}$, we can rewrite our Lagragian with Caputo derivatives of order $0<\alpha+j<\alpha+n$. However, different from \cite{ART} where a similar functional is considered, we adopt the notation ${_a^C D^{\alpha}_t} \frac{d^j}{dt^j}$ because it is more didatic for the present purpose. Finally, it is important to mention that our Action Principle generalizes \cite{Riewe} and differs from the general theorems proposed in \cite{Agrawal,JerkLazo,CressonInizan} (for a review in recent advances in calculus of variations with fractional derivatives see \cite{MalinowskaTorres}).

%%%%%%%%%%%%%%%%%%%%%%%%%%%%%%%%%%%%%%%%%%%%%%%%%%%%%%%%%%%%%%%%%%%%%%%%%%%%%%%%%%%%%%%%%%

\subsection{A quadratic Lagrangian for the classical accelerated point charge}

As an example of application, in this section we obtained a quadratic Lagrangian for the classical accelerated point charge \cite{landau}, where we have a radiation recoil force proportional to $\dddot{x}$. We also show that our Lagrangian is physical in the sense it provides meaningful relations for the momentum and Hamiltonian. Let us consider the Lagrangian
\begin{equation}
\lb{c1}
L\left(x,\dot{x},{_t^C D^{\frac{1}{2}}_b} \dot{x}\right)=\frac{1}{2}m\left(\dot{x}\right)^2-U(x)+\frac{2e^2}{6c^3}\left({_t^C D^{\frac{1}{2}}_b} \dot{x}\right)^2,
\end{equation}
where the three terms in \rf{c1} represent the kinetic energy, potential energy, and the fractional radiation recoil energy, respectively. Since the physical equation of motion is obtained in the limit $a\rightarrow b$, if we consider the last term in \rf{c1} up to first order in $\Delta t=b-a$ we get:
\beq
\lb{c1b}
\frac{2e^2}{6c^3}\left({_t^C D^{\frac{1}{2}}_b} \dot{x}\right)^2\approx \frac{2e^2}{6c^3}\left(\frac{\Gamma(1)}{\Gamma(\frac{3}{2})}\right)^2\left(\ddot{x}\right)^2\Delta t \approx \frac{2}{\pi} \frac{2e^2}{3c^3}\left(\ddot{x}\right)^2\Delta t,
\eeq
that coincide, apart from the multiplicative constant $2/\pi$, with the energy lost in the time interval $\Delta t$ by the radiation recoil \cite{landau}. As in the linear friction problem, this additional constant is a consequence of the use of fractional derivatives in the Lagrangian and do not appears in the equation of motion after we apply the action principle. Furthermore, the Lagrangian \rf{c1} is physical in the sense it provide us with physically meaningful relations for the momentum and the Hamiltonian. If we define the canonical variables
\beq
\lb{c1c}
q_{1}=\dot{x}, \;\;\; q_{\frac{3}{2}}={_t^C D^{\frac{1}{2}}_b} \dot{x},
\eeq
and
\beq
\lb{c1d}
p_1=\frac{\partial L}{\partial q_1}=m\dot{x}, \;\;\; p_{\frac{3}{2}}=\frac{\partial L}{\partial q_{\frac{3}{2}}}=\frac{2e^2}{3c^3}{_t^C D^{\frac{1}{2}}_b} \dot{x},
\eeq
we obtain the Hamiltonian
\beq
\lb{c1e}
H=q_1p_1+q_{\frac{3}{2}}p_{\frac{3}{2}}-L=\frac{1}{2}m\left(\dot{x}\right)^2+U(x)+\frac{2e^2}{6c^3}\left({_t^C D^{\frac{1}{2}}_b} \dot{x}\right)^2.
\eeq
From \rf{c1d} and \rf{c1e} we can see that the Lagrangian \rf{c1} is physical in the sense it provides us a correct relation for the momentum $p_1=m\dot{x}$, and a physically meaningful Hamiltonian (it is the sum of all energies). Furthermore, the additional fractional momentum $p_{\frac{3}{2}}$ goes to zero when we take the limit $a\rightarrow b$.

Finally, by inserting \rf{c1} into our generalized Euler-Lagrange equation \rf{t2} we obtain the following relation
\beq
\lb{c2}
m \ddot{x}+\frac{2e^2}{3c^3}\frac{d}{dt}\left( {_a D^{\frac{1}{2}}_t} {_t^C D^{\frac{1}{2}}_b}\right)\dot{x}=F(x),
\eeq
where $F(x)=-\frac{d}{dx}U(x)$ is the external force. The equation of motion for the accelerated charge is obtained from \rf{c2} by taking the limit $a\rightarrow b$. By taking the limit and using the approximation \rf{b7} and the relation \rf{a18} we obtain
\begin{equation}
\lb{c4}
m \ddot{x}-\frac{2e^2}{3c^3}\dddot{x}=F(x),
\end{equation}
that is the correct equation of motion for the charged particles \cite{landau}.

%%%%%%%%%%%%%%%%%%%%%%%%%%%%%%%%%%%%%%%%%%%%%%%%%%%%%%%%%%%%%%%%%%%%%%%%%%%%%%%%%%%%%%%%%%

\section{Conclusions}

In the present work, we redefined the action principle for dissipative systems and we solved the mathematical inconsistencies present in the Riewe's approach. In order to formulate a quadratic Lagrangian for non-conservative systems, the Lagrangian functions proposed depends on mixed integer order and fractional order derivatives. Our approach has the advantage of enable us to formulate physical Lagrangians for dissipative systems in the sense they provide us a correct relation for the momentum, and a physically meaningful Hamiltonian. Furthermore, our action principle also enable us to formulate Lagrangians for higher-order open and dissipative systems. As examples of applications for non-conservative systems, we formulate a quadratic Lagrangian for a particle under a frictional force proportional to the velocity, and for the accelerated point charge.

%%%%%%%%%%%%%%%%%%%%%%%%%%%%%%%%%%%%%%%%%%%%%%%%%%%%%%%%%%%%%%%%%%%%%%%%%%%%%%%%%%%%%%%%%%

\section*{acknowledgments}
This work was partially supported by CNPq, CAPES and FAPERGS (Brazilian research funding agencies). 

%%%%%%%%%%%%%%%%%%%%%%%%%%%%%%%%%%%%%%%%%%%%%%%%%%%%%%%%%%%%%%%%%%%%%%%%%%%%%%%%%%%%%%%%%%

\appendix

\section{The limit $a\rightarrow b$ for Riemann-Liouville derivatives}

In this appendix we analyze the limit $a\rightarrow b$ for Riemann-Liouville derivatives. We prove the following lemma

\begin{lemma} Let $x(t)\in C^1[a,b]$, and let $\alpha\in \mathbb{R}$ with $0<\alpha<1$. If $x(b)\neq 0$ then
\beq
\lb{ap1}
\lim_{a\rightarrow b^-} {_a D_t^{\alpha}} x(t)=\lim_{a\rightarrow b^-} {_t D_b^{\alpha}} x(t)=\pm \infty,
\eeq
and the limit
\beq
\lb{ap2}
\lim_{a\rightarrow b^-} \frac{{_a D_t^{\alpha}} x(t)}{{_t D_b^{\alpha}}x(t)}
\eeq
is indefinite.
\end{lemma}

\begin{proof}
Since $x(t)\in C^1[a,b]$, there are real numbers $a<c_a<t$ and $t<c_b<b$ such that $x(t)=x(a)+\dot{x}(c_a)(t-a)$ in $[a,t]$ and $x(t)=x(b)-\dot{x}(c_b)(b-t)$ in $[t,b]$. Then, from \rf{a81} and \rf{a82}, we have:
\beq
\lb{ap3}
\lim_{a\rightarrow b^-} {_a D_t^{\alpha}} x(t)=\lim_{a\rightarrow b^-} \left(x(a)\frac{(t-a)^{-\alpha}}{\Gamma(1-\alpha)}+\dot{x}(c_a)\frac{(t-a)^{1-\alpha}}{\Gamma(2-\alpha)}\right)=x(b)\lim_{a\rightarrow b^-}\frac{(t-a)^{-\alpha}}{\Gamma(1-\alpha)},
\eeq
\beq
\lb{ap4}
\lim_{a\rightarrow b^-} {_t D_b^{\alpha}} x(t)=\lim_{a\rightarrow b^-} \left(x(b)\frac{(b-t)^{-\alpha}}{\Gamma(1-\alpha)}-\dot{x}(c_b)\frac{(b-t)^{1-\alpha}}{\Gamma(2-\alpha)}\right)=x(b)\lim_{a\rightarrow b^-} \frac{(b-t)^{-\alpha}}{\Gamma(1-\alpha)},
\eeq
and
\beq
\lb{ap5}
\begin{split}
\lim_{a\rightarrow b^-} \frac{{_a D_t^{\alpha}} x(t)}{{_t D_b^{\alpha}}x(t)}&=\lim_{a\rightarrow b^-} \frac{x(a)\frac{(t-a)^{-\alpha}}{\Gamma(1-\alpha)}+\dot{x}(c_a)\frac{(t-a)^{1-\alpha}}{\Gamma(2-\alpha)}}{x(b)\frac{(b-t)^{-\alpha}}{\Gamma(1-\alpha)}-\dot{x}(c_b)\frac{(b-t)^{1-\alpha}}{\Gamma(2-\alpha)}}=
\lim_{a\rightarrow b^-}\frac{x(a)}{x(b)} \left(\frac{b-t}{t-a}\right)^{\alpha}\\
&=\lim_{a\rightarrow b^-}\left(\frac{b-t}{t-a}\right)^{\alpha}.
\end{split}
\eeq
It can be easily seen that the limits \rf{ap3} and \rf{ap4} are divergent since $x(b)\neq 0$. On the other hand, the limit \rf{ap5} is indefinite since $t$ is any arbitrary number in $[a,b]$ (see discussion on section $3$).
\end{proof}

\end{document}